%% file: TFT.tex
\documentclass[12pt]{article}

\usepackage{wrapfig, amsmath, amsthm, amssymb, fullpage, placeins}

\usepackage{subfigure}
\usepackage{array}
\usepackage[ruled,linesnumbered,vlined]{algorithm2e}

\include{header}

\theoremstyle{theorem}
\newtheorem{theorem}{Theorem}

\theoremstyle{definition}
\newtheorem*{definition}{Definition}
\newtheorem*{remark}{Remark}
\newtheorem*{example}{Example}

\newcommand{\RR}{\mathcal{R}}
\newcommand{\DFT}{\text{DFT}}
\newcommand{\TFT}{\text{TFT}}
\newcommand{\w}{\omega}
\renewcommand{\a}{\mathbf{a}}
\newcommand{\xs}[1]{\mathbf{x_{#1}}}

\title{An Illustrated Introduction to the Truncated Fourier Transform}
\author{Paul Vrbik. \\
School of Mathematical and Physical Sciences \\
The University of Newcastle\\
Callaghan, Australia\\
\tt{paulvrbik@gmail.com}}
\date{\today}

\begin{document}
\maketitle

\begin{abstract}
The Truncated Fourier Transform (\textsc{tft}) is a variation of the Discrete Fourier Transform (\textsc{dft}/\textsc{fft}) that allows for input vectors that do \emph{not} have length $2^n$ for $n$ a positive integer. We present the univariate version of the \textsc{tft}, originally due to Joris van der Hoeven, heavily illustrating the presentation in order to make these methods accessible to a broader audience.
\end{abstract}


\section{Introduction}

In 1965 Cooley and Tukey developed the \emph{Discrete Fourier Transform} (\textsc{dft}) to recover continuous functions from discrete samples. This was a landmark discovery because it allowed for the digital manipulation of analogue signals (like sound) by computers. Soon after a variant called the \emph{Fast Fourier Transform} (\textsc{fft}) eclipsed the \textsc{dft} to the extent that \textsc{fft} is often mistakenly substituted for \textsc{dft}. As its name implies, the \textsc{fft} is a method for computing the \textsc{dft} faster.



\textsc{fft}s have an interesting application in Computer Algebra. Let $\RR$ be a ring with $2 \in \RR$ a unit. If $\RR$ has a primitive $n$th root of unity $\w$ with $n=2^p$ (i.e. $\w^{n/2}=-1$) then the \textsc{fft} computes the product of two polynomials $P,Q \in \RR[x]$ with $\deg(PQ)<n$ in $O(n\log n)$ operations in $\RR$. Unfortunately, when $\deg(PQ)$ is sufficiently far from a power of two \emph{many} computations wasted. This deficiency was addressed by the signal processing community using a method called \textsc{fft}-pruning \cite{SigProcessing}. However, the difficult inversion of this method is due to van der Hoeven \cite{TFT2}\cite{TFT1}.\smallskip

In Section \ref{Section::DFT} we outline the \textsc{dft}, including a method for its non-recursive implementation. In Section \ref{Section::TFT} we develop the ``pruned'' variant, called Truncated Fourier Transform (\textsc{tft}). Finally, in Section \ref{Section::InvTFT}, we show how the \textsc{tft} can be inverted and outline the algorithm for doing so in Section \ref{Section::AlgTFT}.

\section{The Discrete Fourier Transform}\label{Section::DFT}
For this paper let $\RR$ be a ring with $2 \in \RR$ a unit and $\w \in \RR$ an $n$th root of unity. The Discrete Fourier Transform\footnote{In signal processing community this is called the ``decimation-in-time'' variant of the \textsc{fft}.}, with respect to $\w$, of vector 
$\a=(a_0, \ldots, a_{n-1}) \in \RR^n$
is the vector $\hat \a = (\hat a_0, \ldots, \hat a_{n-1}) \in \RR^n$ with $$ \hat a_i = \sum_{j=0}^{n-1} a_j \w^{ij}.$$

Alternatively we can view these $n$-tuples as encoding the coefficients of polynomials from $\RR[x]$ and define the \textsc{dft} with respect to $\w$ as the mapping 
\begin{align*}
\DFT_{\w}:  \RR[x] &\to \RR^n \\
A(x) = a_0 + \cdots + a_{n-1} x^{n-1} &\mapsto (A(\w^0), \ldots, A(\w^{n-1})).
\end{align*}
We let the relationship between $A$ and its coefficients be implicit and write 
$$\DFT_{\w}\brac{a_0,\ldots,a_{n-1}} := (A(\w^0), \ldots, A(\w^{n-1}))$$
when $A = a_0 + \cdots + a_{n-1} x^{n-1}$.

The \textsc{dft} can be computed efficiently using binary splitting. This method requires evaluation only at $\w^{2^i}$ for $i \in \{0,\ldots,p-1\}$, rather than at all $\w^0, \ldots, \w^{n-1}$. To compute the \textsc{dft} of $A$ with respect to $\w$ we write
$$(b_0,c_0,\ldots,b_{n/2-1},c_{n/2-1}) := \brac{a_0, \ldots, a_{n-1}}$$
and recursively compute the \textsc{dft} of $(b_0,\ldots,b_{n/2-1})$ and $(c_0,\ldots,c_{n/2-1})$ with respect to $\w^2$:
\begin{align*}
(\hat b_0 , \ldots, \hat b_{n/2-1}) := \DFT_{\w^2}(b_0,\ldots,b_{n/2-1}) , \\
(\hat c_0 , \ldots, \hat c_{n/2-1}) := \DFT_{\w^2}(c_0,\ldots,c_{n/2-1}) .
\end{align*}
Finally, we construct $\hat \a$ according to
\begin{align*}
\DFT_\omega (a_0, \ldots, a_{n-1}) = (& \hat b_0 + \hat c_0 , \ldots, \hat b_{n/2-1} + \hat c_{n/2-1}\omega^{n/2-1}, \\
& \qquad \hat b_0 - \hat c_0 , \ldots, \hat b_{n/2-1} - \hat c_{n/2-1}\omega^{n/2-1} ).
\end{align*}

This description has a natural implementation as a recursive algorithm, but in practice it is often more efficient to implement an in-place algorithm that eliminates the overhead of creating recursive stacks.

\begin{definition}
Let $i$ and $p$ be a positive integers and let $i = i_0 2^0 + \cdots + i_p 2^p$ for $i_0,\ldots,i_p \in \cbrac{0,1}$. The length-$p$ bitwise reverse of $i$ is given by
$$ [i]_p := i_p 2^0 + \cdots + i_{0} 2^p.$$
\end{definition}

\begin{example} $[3]_5 = 24$ and $[11]_5 = 26$ because
\begin{align*}
[3]_5 &= [1\cdot 2^0 + 1\cdot 2^1 + 0 \cdot 2^2 + 0 \cdot 2^3 + 0 \cdot 2^4]_5 \\
      &= 0 \cdot 2^0 + 0\cdot 2^1 + 0 \cdot 2^2 + 1 \cdot 2^3 + 1 \cdot 2^4\\
      &= 24
\end{align*}
and
\begin{align*}
[11]_5 &= [1\cdot 2^0 + 1\cdot 2^1 + 0 \cdot 2^2 + 1 \cdot 2^3 + 0 \cdot 2^4]_5 \\
      &= 0 \cdot 2^0 + 1\cdot 2^1 + 0 \cdot 2^2 + 1 \cdot 2^3 + 1 \cdot 2^4\\
      &= 26.
\end{align*} 
Notice if we were to write 3, 24, 11, and 26 as a binary numbers to five digits we have 00011 reverses to 11000 and 01011 reverses to 11010 --- in fact this is the inspiration for the name ``bitwise reverse.''
\end{example}

For the in-place non-recursive \textsc{dft} algorithm, we require \emph{only one vector} of length $n$. Initially, at step zero, this vector is
$$\xs{0} = (x_{0,0}, \ldots, x_{0,n-1}) := (a_0, \ldots, a_{n-1})$$
and is updated (incrementally) at steps $s \in \{ 1 , \ldots, p \}$ by the rule
\begin{equation}
\left[\begin{array}{c}x_{s,im_s+j} \\x_{s,(i+1)m_s+j}\end{array}\right]:=\left[\begin{array}{cr}1 & \omega^{[i]_sm_s} \\1 & -\omega^{[i]_sm_s}\end{array}\right]\left[\begin{array}{c}x_{s-1,im_s+j} \\x_{s-1,(i+1)m_s+j}\end{array}\right]
\label{EqButterfly}
\end{equation}
where $m_s = 2^{p-s}$ and for all $i \in \{0, 2, \ldots, n/m_s-2\}$,  $j \in \{0,\ldots,m_s-1\}$. Note that two additions and \emph{one} multiplication are done in (\ref{EqButterfly}) as one product is merely the negation of the other.\medskip

We illustrate the dependencies of the $x_{s,i}$ values in (\ref{EqButterfly}) with 
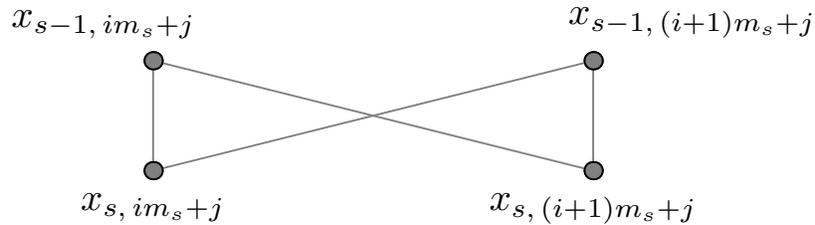
\begin{center}
   \includegraphics[width=0.65\textwidth]{tst.35}
\end{center}
We call this a ``butterfly'' after the shape it forms and may say $m_s$ controls the width --- the value of which decreases as $s$ increases. By placing these butterflies on a $s \times n$ grid  (Figure \ref{Butterfly}) we can see which values of $\xx_s$ are required to compute particular entires of $\xx_{s+1}$ (and vice-versa). For example
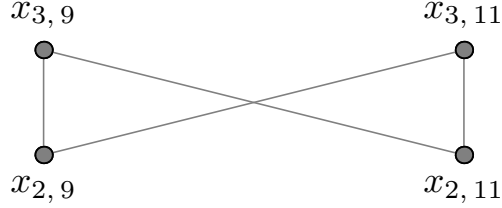
\begin{center}
   \includegraphics[width=0.4\textwidth]{tst.36}\;\;\;\;\;\;\;\;
\end{center}
denotes that $x_{3,\,9}$ and $x_{3,\,11}$ are required to determine $x_{2,\,9}$ and $x_{2,\,11}$ (and vice-versa).\medskip

\begin{figure}[htbp]
\begin{center}
\includegraphics[width=\textwidth]{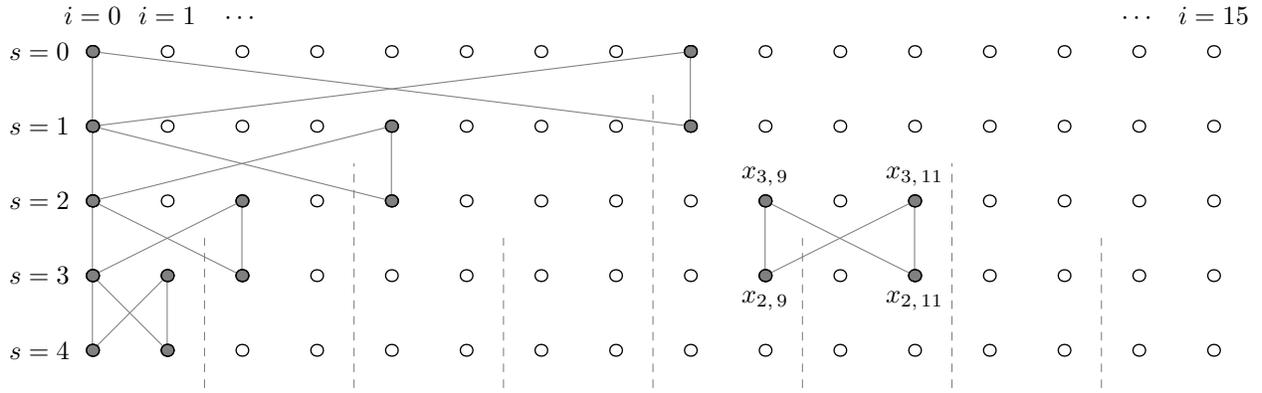}
\caption{Schematic representation of Equation (\ref{EqButterfly}) using butterflies to illustrate the value dependencies at various steps $s$. The grid has rows $s=0,\ldots,4$ and columns $n=0,\ldots,15$ for illustrating the $x_{s,i}$ values.}
\label{Butterfly}
\end{center}
\end{figure}

Using induction over $s$,
$$x_{s,im_s+j} = \brac{ \DFT_{\w^{m_s}}(a_j, a_{m_s+j}, \ldots, a_{n-m_s+j}) }_{[i]_s},$$
for all $i \in \{ 0 , \ldots, n/m_s -1 \}$ and $j \in \{0, \ldots, m_s -1 \}$ \cite{TFT2}. In particular, when $s=p$ and $j=0$, we have
$$
x_{p,i} = \hat a_{[i]_p} \;\text{ and }\; \hat a_i = x_{p,[i]_p}
$$
for all $i \in \{ 0, \ldots, n-1 \}$. That is, $\hat \a$ is a (specific) permutation of $\xs{p}$ as illustrated in Figure \ref{FFT1}.

The key property of the \textsc{dft} is that it is straightforward to invert, that is to recover $\mathbf{a}$ from $\hat{\mathbf{a}}$:
\begin{equation} \label{invFFT}
\DFT_{\omega^{-1}}( \hat{\mathbf{a}})_i = \DFT_{\omega^{-1}}(\DFT_\omega{ (\mathbf{a}}))_i = \sum_{k=0}^{n-1} \sum_{j=0}^{n-1} a_i \omega^{(i-k)j} = n a_i
\end{equation}
since $\sum_{j=0}^{n-1} \omega^{(i-k)j} = 0$ whenever $i \neq k$. This yields a polynomial multiplication algorithm that does $O(n \log n)$ operations in $\RR$ (see \cite[\S4.7]{AFCA} for the outline of this algorithm).

\begin{figure}[htbp]
\begin{center}
\includegraphics[width=\textwidth]{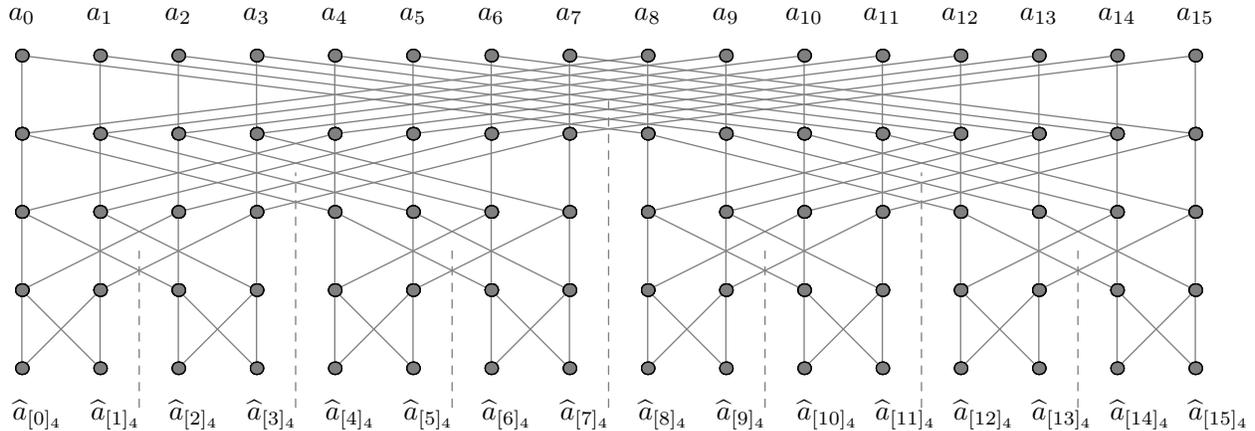}
\caption{The Discrete Fourier Transform for $n=16$. The top row, corresponding to $s=0$, are the initial values $\a$. The bottom row, corresponding to $s=4$, is a permutation of $\hat \a$ (the result of the DFT on $\a$).}
\label{FFT1}
\end{center}
\end{figure}

\FloatBarrier

\section{The Truncated Fourier Transform}\label{Section::TFT}
The motivation behind the Truncated Fourier Transform (\textsc{tft}) is the observation that many computations are wasted when the length of $\a$ (the input) is not a power of two.\footnote{The \textsc{tft} is exactly equivalent to a technique called ``\textsc{fft} pruning'' in the signal processing literature \cite{SigProcessing}.} This is entirely the fault of the strategy where one ``completes'' the $\ell$-tuple $\a = (a_0,\ldots,a_{\ell-1})$ by setting $a_{i} = 0$ when $i \geq \ell$ to artificially extend the length of $\a$ to the nearest power of two (so the \textsc{dft} can be executed as usual).

However, despite the fact that we may only want $\ell$ components of $\hat \a$, the \textsc{dft} will calculate \emph{all} of them. Thus computation is wasted. We illustrate this in Figures \ref{OverFFT} and \ref{TFT0}. This type of wasted computation is relevant when using the \textsc{dft} to multiply polynomials --- their products are rarely of degree one less some power of two.

The definition of the \textsc{tft} is similar to that of the \textsc{dft} with the exception that the input and output vector ($\a$ resp.\ $\hat \a$) are not necessarily of length some power of two.  More precisely the \textsc{tft} of an $\ell$-tuple $(a_0,\ldots,a_{\ell-1}) \in \RR^\ell$ is the $\ell$-tuple $$\brac{ A( \w^{[0]_p}), \ldots, A(\w^{[\ell-1]_p}) } \in \RR^\ell.$$
where $n=2^p$, $\ell < n$ (usually $\ell \geq n/2$) and $\omega$ a $n$th root of unity.

\begin{figure}[htbp]
\begin{center}
\includegraphics[width=\textwidth]{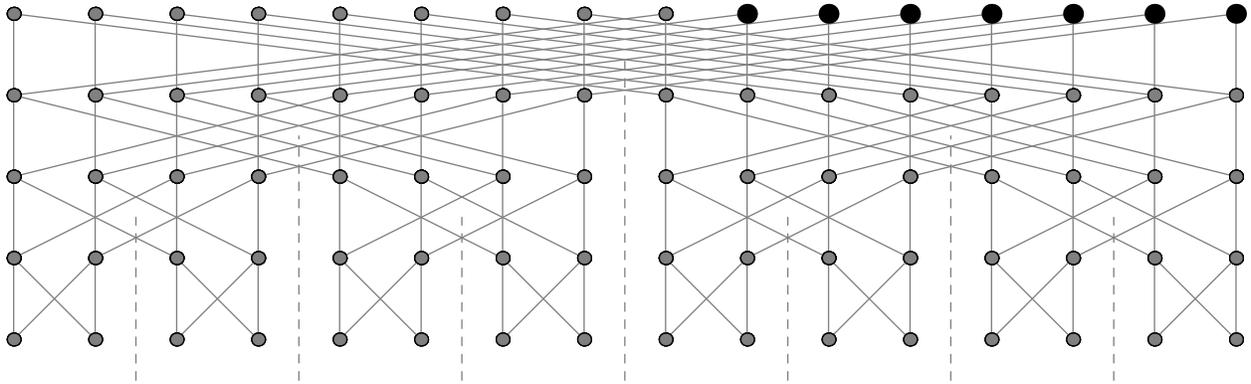}
\caption{The \textsc{dft} with ``artificial'' zero points (large black dots).}
\label{OverFFT}
\end{center}
\end{figure}

\begin{figure}[htbp]
\begin{center}
\includegraphics[width=\textwidth]{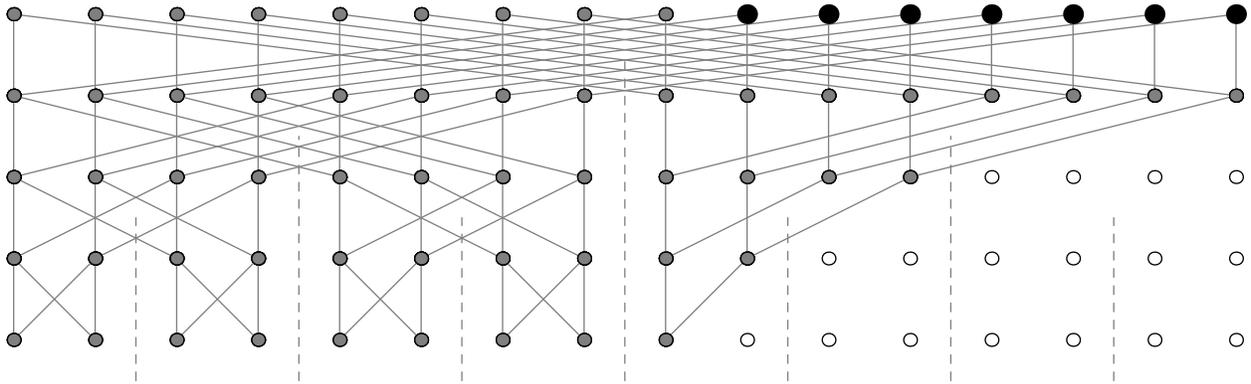}
\caption{Removing all unnecessary computations from Figure \ref{OverFFT} gives the schematic representation of the \textsc{tft}.}
\label{TFT0}
\end{center}
\end{figure}

\begin{remark}
van der Hoeven \cite{TFT2} gives a more general description of the \textsc{tft} where one can choose an initial vector $(x_{0,i_0}, \ldots, x_{0,i_n})$ and target vector $(x_{p,j_0}, \ldots, x_{p,j_n})$. Provided the $i_k$'s are distinct one can carry out the \textsc{tft} by considering the full \textsc{dft} and removing all computations not required for the desired output. In this paper, we restrict our discussion to that of the scenario in Figure \ref{TFT0} (where the input and output are the same initial segments) because it can be used for polynomial multiplication, and because it yields the most improvement. $\diamond$
\end{remark}

If we only allow ourselves to operate in a size $n$ vector it is straightforward to modify the in-place \textsc{dft} algorithm from the previous section to execute the \textsc{tft}. (It should be emphasized that this only saves computation and not space. For a ``true'' in-place \textsc{tft} algorithm that operates in an array of size $\ell$, see Harvey and Roche's \cite{InPlaceTFT}.) At stage $s$ it suffices to compute 
$(x_{s,0}, \ldots, x_{s,j})$ with $j=\lceil \ell/m_s \rceil m_s-1$ where $m_s = 2^{p-s}$.\footnote{This is a correction to the bound given in \cite{TFT1} as pointed out in \cite{TFT2}.}

\begin{theorem}
Let $n=2^p$, $1 \leq \ell < n$ and $\w \in \RR$ be a primitive $n$th root of unity in $\RR$. The \textsc{tft} of an $\ell$-tuple $(a_0, \ldots, a_{\ell-1})$ with respect to $\w$ can be computed using at most $\ell p + n$ additions and $\lfloor (\ell p + n)/2 \rfloor$ multiplications of powers of $\w$.
\end{theorem}

\begin{proof}
Let $j=(\lceil \ell_s/m_s \rceil)m_s-1$; at stage $s$ we compute $(x_{s,0}, \ldots, x_{s,j})$. So, in addition to $x_{s,0}, \ldots, x_{s,\ell-1}$ we compute $$(\lceil \ell/m_s \rceil)m_s -1 - \ell \leq m_s$$ more values. Therefore, in total, we compute at most 
\begin{align*}
p\ell + \sum_{s=1}^{p} {m_s} 
&= p \ell + 2^{p-1} + 2^{p-2} + \cdots + 1 = p\ell + 2^p - 1 < p \ell + n
\end{align*}
values $x_{s,i}$. The result follows.
\end{proof}

\section{Inverting The Truncated Fourier Transform}\label{Section::InvTFT}
Unfortunately, the \textsc{tft} cannot be inverted by merely doing another \textsc{tft} with $\w^{-1}$ and adjusting the output by some constant factor like inverse of the \textsc{dft}. Simply put: we are missing information and must account for this.

\begin{example}
Let $\RR = \ZZ / 13\ZZ$, $n=2^2=4$, with $\w = 5$ a $n$th primitive root of unity. Setting $A(x) = a_0+a_1x+a_2x^2$, the \textsc{tft} of $\a=(a_0,a_1,a_2)$ at 5 is
$$
\sbrac{\begin{array}{c}
A(\w^0) \\
A(\w^2) \\
A(\w^1)
\end{array}}
= 
\sbrac{\begin{array}{c}
A(1) \\
A(-1) \\
A(5)
\end{array}}
=
\sbrac{\begin{array}{c}
a_0 + a_1 + a_2 \\
a_0 - a_1 + a_2 \\
a_0 + 5a_1 - a_2
\end{array}}.
$$
Now, to show the \textsc{tft} of this with respect to $\w^{-1}$ is \emph{not} $\a$, define 
$$\mathbf{b}
= 
\sbrac{\begin{array}{c}
b_0 \\
b_1 \\
b_2
\end{array}}
=
\sbrac{\begin{array}{c}
a_0 + a_1 +a_2 \\
a_0-a_1+a_2 \\
a_0 + 5a_1 - a_2
\end{array}}.
$$
The \textsc{tft} of $\mathbf{b}$ with respect to $\w^{-1} = -5$ is
$$
\sbrac{\begin{array}{c}
B \brac{\w^0} \\
B \brac{ \w^{-2}} \\
B \brac{ \w^{-1}} 
\end{array}}=
\sbrac{\begin{array}{c}
B(1) \\
B(-1) \\
B(5)
\end{array}}=
\sbrac{\begin{array}{c}
b_0+b_1+b_2 \\
b_0-b_1+b_2 \\
b_0 -5b_1 -b_2
\end{array}}=
\sbrac{\begin{array}{c}
3a_0 + 5a_1 + a_2 \\
a_0 - 6a_1 - a_2 \\
-5a_0 + a_1 - 3a_2
\end{array}}
$$
which is \emph{not} a constant multiple of $\TFT_\w (\a)$.

This discrepancy is caused by the completion of $\mathbf{b}$ to $(b_0, b_1, b_2, 0)$ --- we should have instead completed $b$ to $(b_0, b_1, b_2,  A(-5) )$.
\end{example}

To invert the \textsc{tft} we 
use the fact that whenever two values among $$x_{s,im_s+j}, x_{s-1,im_s+j} \;\text{ and }\; x_{s,(i+1)m_s+j},x_{s-1,(i+1)m_s+j}$$ are known, that the other values can be deduced. That is, if two values of some butterfly are known then the other two values can be calculated using (\ref{EqButterfly}) as the relevant matrix is invertible. Moreover, these relations only involve shifting (multiplication and division by two), additions, subtractions and multiplications by roots of unity --- an ideal scenario for implementation.

As with \textsc{dft}, observe that $x_{p-k,0}, \ldots, x_{p-k,2^k-1}$ can be calculated from $x_{p,0}, \ldots, x_{p,2^k-1}$. This is because all the butterfly relations necessary to move up like this never require $x_{s,2^k+j}$ for any $s \in \{p-k, \ldots, p\}$ and $j>0$. This is illustrated in Figure \ref{boxes}. More generally, we have that
$$x_{p,2^j+2^k}, \ldots, x_{p,2^j + 2^k -1}$$
is sufficient information to compute
$$x_{p-k,2^j}, \ldots, x_{p-k,2^j +2^k -1 }$$ 
provided that $0<k\leq j < p$. (In Algorithm 1, that follows, we call this a ``self-contained push up''.)\\

\begin{figure}[h]
\begin{center}
\includegraphics[width=\textwidth]{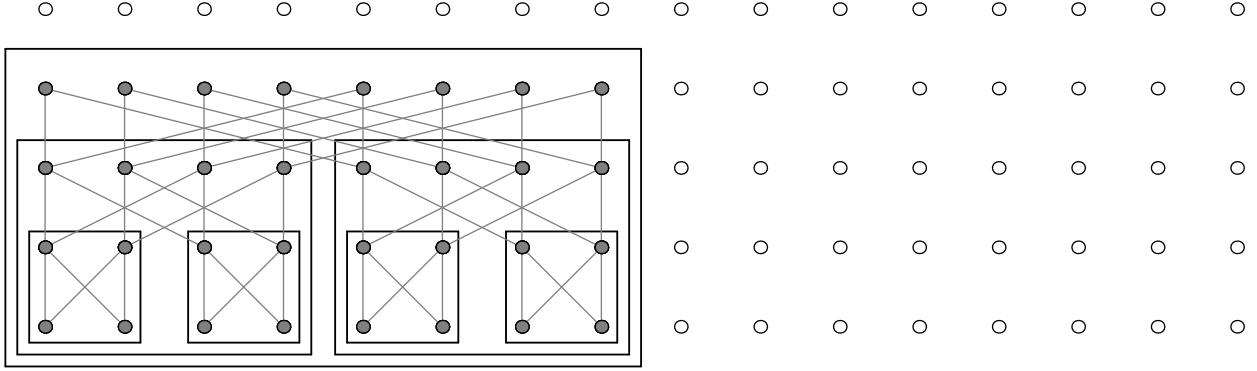}
\caption{The computations in the boxes are self contained.}
\label{boxes}
\end{center}
\end{figure}

\section{Inverse \textsc{tft} Algorithm}\label{Section::AlgTFT}

Finally, in this section, we present a simple recursive description of the inverse \textsc{tft} algorithm for the case we have restricted ourselves to (all zeroes packed at the end). The algorithm operates in a length $n$ array $\xx = \brac{\xx_0, \ldots, \xx_{n-1}}$ for which we assume access; here $n=2^p$ corresponds to $\omega$, a $n$th primitive root of unity. Initially, the content of the array is
$$\xx := \brac{x_{p,\,0}, \ldots, x_{p,\,\ell-1},\, 0, \ldots,0}$$
where $\brac{x_{p,\,0}, \ldots, x_{p,\,\ell-1}}$ is the result of the \textsc{tft} on $\brac{x_{0,\,0}, \ldots, x_{0,\,\ell-1},\,0,\ldots,0}$.

In keeping with our ``Illustrated'' description we use pictures, like Figure \ref{PicPic},
to indicate what values are known (solid dots $\bullet$) and what value to calculate (empty dot $\circ$). For instance, ``push down $\vec{x}_k$ with Figure \ref{PicPic}'', is shorthand for: use $\xx_k = x_{s-1,\,im_s+j}$ and $\xx_{k+m_s+j}= x_{s-1,\,(i+1)m_s+j}$ to determine $x_{s,\,im_s+j}$. We emphasize with an arrow that this new value should also overwrite the one at $\xx_k$. This calculation is easily accomplished using (\ref{EqButterfly}) with a caveat: the values $i$ and $j$ are not explicitly known. What \emph{is} known is $s$, and therefore $m_s$, and some array position $k$. Observe that $i$ is recovered by $i = k \text{ quo } m_s$ (the quotient of $k/m_s$).
\begin{figure}[htp]
  \begin{center}
    \includegraphics[width=0.7\textwidth]{tst.27}
  \end{center}
  \caption{Overwrite $\xx[im_s+j]$ with $\xx[im_s+j] + \w^{[i]_sm_s}\xx[(i+1)m_s+j]$.}
\label{PicPic}
\end{figure}

The full description of the inverse \textsc{tft} follows in Algorithm \ref{InvTFT}; note that the initial call is {\bf InvTFT}$\brac{0,\, \ell-1 ,\, n-1,\,1}$. A visual depiction of Algorithm \ref{InvTFT} is given in Figure \ref{BigInvTFT}. A sketch of a proof of its correctness follows.

\input{algorithms}

\begin{theorem}
Algorithm \ref{InvTFT}, initially called with ${\rm \bf InvTFT}\brac{0,\, \ell-1 ,\, n-1,\,1}$ and given access to the zero-indexed length $n$ array
\begin{equation}\label{forward}
\xx = \brac{x_{p,\,0}, \ldots, x_{p,\,\ell-1},\, 0, \ldots,0}
\end{equation}
will terminate with
\begin{equation}\label{backward}
\xx = \brac{x_{0,\,0}, \ldots, x_{0,\,\ell-1},\, 0, \ldots,0}
\end{equation}
where (\ref{forward}) is the result of the \textsc{tft} on (\ref{backward}).
\end{theorem}

\begin{proof}[Termination]
Let $\head_i$, $\tail_i$, and $\last_i$ be the values of head, tail, and last at the $i$th recursive call. Consider the {integer} sequences given by
\begin{align*}
\alpha_i &= \tail_i - \head_i \in \ZZ, \\ 
\beta_i &= \tail_i - \left\lfloor \dfrac{\last_i-\head_i}{2} + \head \right\rfloor \in \ZZ.
\end{align*}

If $\head_i > \tail_i$ then we have termination. Otherwise, either branch (7) executes giving
\begin{align*}
\head_{i+1} &= \floor{ \dfrac{\last_i - \head_i}{2} } + \head_i > \head_i\\
\tail_{i+1} &= \tail_i \\
\last_{i+1} &= \last_i
\end{align*}
and thus $\alpha_{i+1} < \alpha_i$, or branch (13) executes, giving
\begin{align*}
\head_{i+1} &= \head_i \\
\tail_{i+1} &= \tail_i \\
\last_{i+1} &= \floor{ \dfrac{\last_i - \head_i}{2} } + \head_i > \head_i .
\end{align*}
and thus $\beta_{i+1} < \beta_i$.

Neither branch can run forever since $\alpha<0$ causes termination and $\beta<0$ means either $\alpha$ strictly decreases or condition (13) fails, forcing termination.
\end{proof}

\begin{proof}[Sketch of correctness]
Figure \ref{MoreHalf} and Figure \ref{LessHalf} demonstrate that self contained regions can be exploited to obtain the  initial values required to complete the inversion. That is to say, for $n \in \{0,\, \ldots,\, p-1 \}$, that
$$
\xx = \brac{x_{p-n-1,\,0}, \ldots, x_{p-n-1,\,2^{n+1}-1} }
$$
can always be calculated from
$$
\xx = \brac{x_{p,\,0}, \ldots, x_{p,\,\ell-1} ,\, x_{p-n-1,\,{\rm tail}+1}, \ldots,\,  x_{p-n-1,\,2^{n+1}-1}}.
$$
\end{proof}

{\center
\begin{figure}[p]
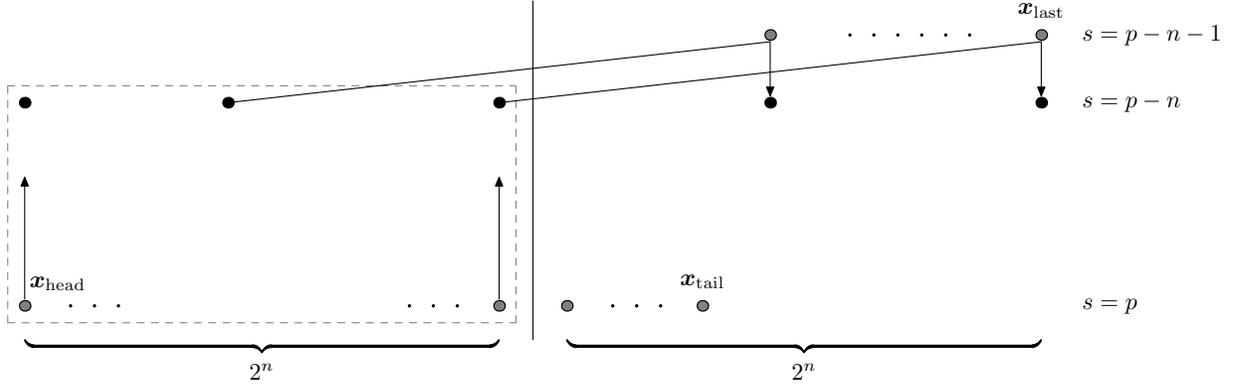

\caption{$\tail \geq {\rm LeftMiddle}$ (i.e. at least half the values are at $x=p$).}
\label{MoreHalf}
\subfigure[Line (8): push up the self contained (dashed) region. This yields values sufficient to push down at line (9). ]{
\includegraphics[width=.98\textwidth]{tst.28}
}
\subfigure[This enables us to make a recursive call on the dashed region (line (12)). By our induction hypothesis this brings all points at $s=p$ to $s=p-n$.]{
\includegraphics[width=.98\textwidth]{tst.29}
}
\subfigure[Sufficient points at $s=p-n$ are known to move to $s=p-n-1$ at line (13).]{
\includegraphics[width=.98\textwidth]{tst.30}
}
\end{figure}
}

{\center
\begin{figure}[p]
\begin{center}
\label{LessHalf}
\subfigure[Initially there is sufficient information to push down at line (14).]{
\includegraphics[width=.98\textwidth]{tst.31}
}
\subfigure[This enables us to make the prescribed recursive call at line (15). ]{
\includegraphics[width=.98\textwidth]{tst.32}
}
\subfigure[By the induction hypothesis this brings the values in the dashed region to $s=p-n$, leaving enough information to move up at line (16).]{
\includegraphics[width=.98\textwidth]{tst.33}
}
\end{center}
\caption{$\tail < {\rm LeftMiddle}$ (i.e. less than half the values are at $x=p$).}
\end{figure}
}



\begin{figure}[p]
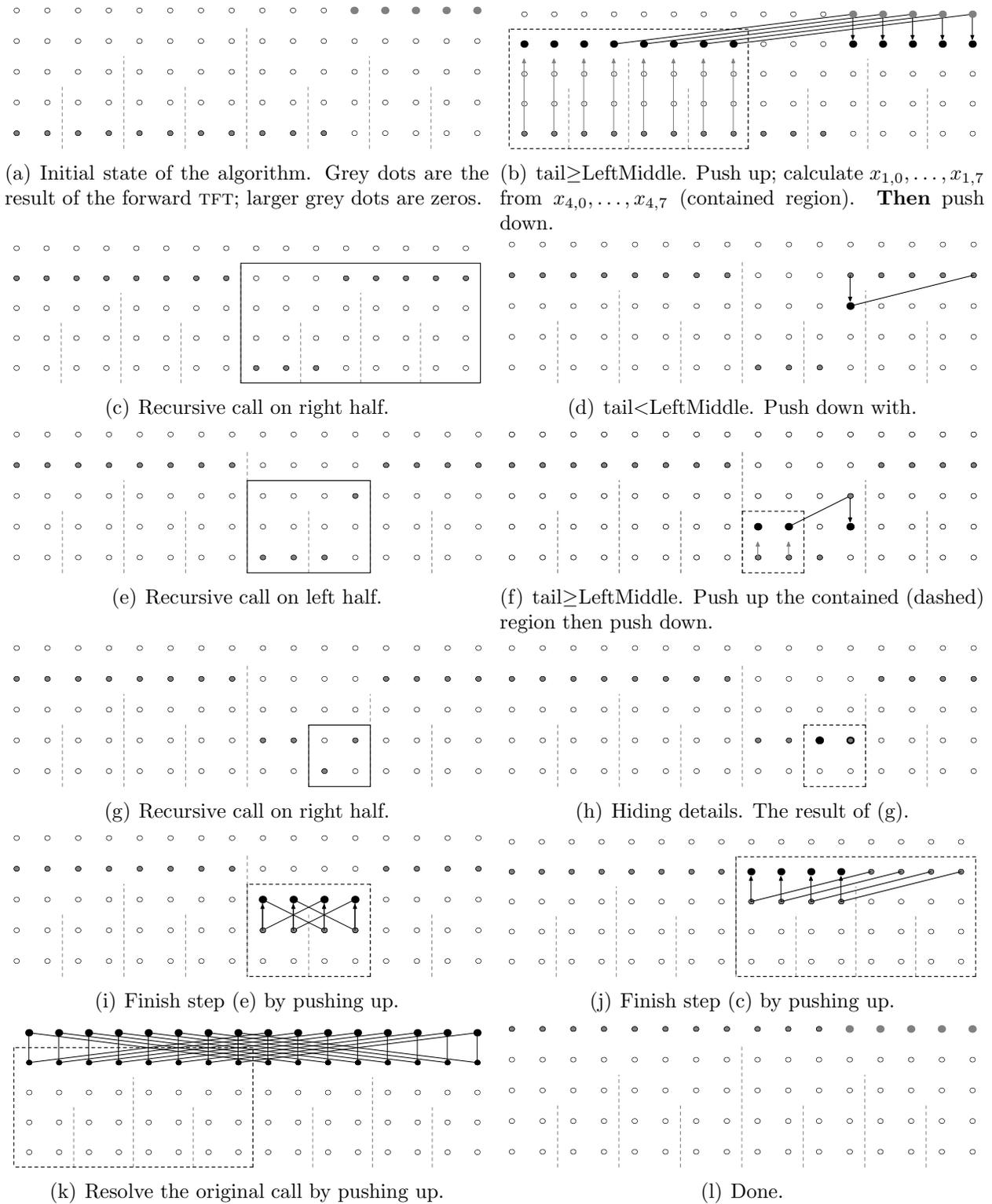

\subfigure[Initial state of the algorithm. Grey dots are the result of the forward \textsc{tft}; larger grey dots are zeros.]{
\includegraphics[width=.48\textwidth]{tst.10}
\label{fig:6a}
}~
\subfigure[tail$\geq$LeftMiddle. Push up; calculate $x_{1,0},\ldots,x_{1,7}$ from $x_{4,0},\ldots,x_{4,7}$ (contained region). \textbf{Then} push down.]{
\includegraphics[width=.48\textwidth]{tst.11}
\label{fig:6b}
}
\subfigure[Recursive call on right half.]{
\includegraphics[width=.48\textwidth]{tst.12}
\label{fig:6c}
}~
\subfigure[tail$<$LeftMiddle. Push down with.]{
\includegraphics[width=.48\textwidth]{tst.13}
\label{fig:6d}
}
\subfigure[Recursive call on left half.]{
\includegraphics[width=.48\textwidth]{tst.14}
\label{fig:6e}
}~
\subfigure[tail$\geq$LeftMiddle. Push up the contained (dashed) region then push down.]{
\includegraphics[width=.48\textwidth]{tst.15}
\label{fig:6f}
}
\subfigure[Recursive call on right half.]{
\includegraphics[width=.48\textwidth]{tst.16}
\label{fig:6g}
}~
\subfigure[Hiding details. The result of (g).]{
\includegraphics[width=.48\textwidth]{tst.18}
\label{fig:6h}
}
\subfigure[Finish step (e) by pushing up.]{
\includegraphics[width=.48\textwidth]{tst.19}
\label{fig:6i}
}~
\subfigure[Finish step (c) by pushing up.]{
\includegraphics[width=.48\textwidth]{tst.20}
\label{fig:6j}
}
\subfigure[Resolve the original call by pushing up.]{
\includegraphics[width=.48\textwidth]{tst.21}
\label{fig:6k}
}~
\subfigure[Done.]{
\includegraphics[width=.48\textwidth]{tst.22}
\label{fig:6l}
}
\caption[Optional caption for list of figures]{Schematic representation of the recursive computation of the inverse \textsc{tft} for $n=16$ and $\ell=11$. }
\label{BigInvTFT}
\end{figure}

\section{Conclusions}
The Truncated Fourier Transform is a novel and elegant way to reduce the number of computations of a \textsc{dft}-based computation by a possible factor of two (which may be significant). Additionally, with the advent of Harvey and Roche's paper \cite{InPlaceTFT}, it is possible to save as much space as computation. The hidden ``cost'' of working with the \textsc{tft} algorithm is the increased difficulty of determining the inverse \textsc{tft}. Although in most cases this is still less costly than the inverse \textsc{dft}, the algorithm is no doubt more difficult to implement.

\section*{Acknowledgements}
The author wishes to thank Dr.\ Dan Roche and Dr.\ \'Eric Schost for reading a draft of this paper and offering suggestions. 

\bibliographystyle{plain}
\bibliography{allbib}

\end{document}

%% file: header.tex
\usepackage{rotating,graphics,psfrag,epsfig,subfigure, color}
\usepackage{fancyvrb}

\usepackage{listings}
\lstdefinelanguage{maple}
	{morekeywords={true, false, try, catch, return, break, error, 
	               module, export, local, option, in, use,
                 and, or, not, xor, xnor,
                 if, then, elif, else, fi,
                 while, for, from, by, to, do, od,
                 proc, nargs, local, global, end, NULL}}
\usepackage[usenames,dvipsnames]{xcolor}
\definecolor{MapleRed}{rgb}{1,0,0}
\definecolor{MapleBlue}{rgb}{0,0,1}
\definecolor{MaplePink}{rgb}{1,0,1}

\renewcommand{\vec}{\boldsymbol}

\theoremstyle{definition}

\newcommand{\brac}[1]{\left( #1 \right)}
\newcommand{\cbrac}[1]{\left\{ #1 \right\}}
\newcommand{\sbrac}[1]{\left[ #1 \right]}

\newcommand{\x}{\boldsymbol{x}}

\newcommand{\xx}{\boldsymbol{x}}

\newcommand{\ZZ}{\mathbb{Z}}

\newcommand{\floor}[1]{\left\lfloor #1 \right\rfloor}

\DefineVerbatimEnvironment%
{MapleInputs}{Verbatim}{formatcom=\color{MapleRed}}

%% file: algorithms.tex

\newcommand{\head}{\text{head}}
\newcommand{\tail}{\text{tail}}
\newcommand{\last}{\text{last}}
\newcommand{\LM}{\text{LeftMiddle}}
\newcommand{\RM}{\text{RightMiddle}}

\begin{algorithm}

\setlength{\parskip}{0.65ex}

\caption{\textbf{InvTFT}$(\head,\, \tail,\, \last,\, s)$}
\label{InvTFT}

\SetKw{Break}{break}\SetKw{AND}{and}\SetKwInOut{Input}{Initial call}\SetKwInOut{Output}{Output\,}
\Input{ InvTFT$\brac{0,\, \ell-1 ,\, n-1,\,1}$; }

\BlankLine
middle $\gets \dfrac{ \last - \head}{2} + \head;$\\[0.5em]

$\LM  \;\,\, \gets \lfloor {\rm middle} \rfloor$\;
$\RM \gets \LM + 1;$\\[1em]

\If {\rm $\head > \tail$} { 
	Base case---do nothing\;
	\Return null;
}

\ElseIf {\rm $\tail \geq \LM$} {
	Push up the self-contained region $\xx_\head$ to $\xx_\LM$\;
	
	\mbox{Push down $\xx_{\tail+1}$ to $\xx_{\last}$ with }\raisebox{-,9em}{\includegraphics[width=.06\textwidth]{tst.24}}~\;
	
	{\bf InvTFT}$\brac{\RM,\, \tail,\, \last,\, s+1}$\;
	$s \gets p - \log_2 \brac{\LM - \head + 1}$\;

	\mbox{Push up (in pairs) $\brac{\xx_\head,\, \xx_{\head+m_s}}$ to $\brac{\xx_\LM,\, \xx_{\LM+m_s}}$} 		\mbox{with }\raisebox{-.9em}{\includegraphics[width=.07\linewidth]{tst.25}}~\;

}
\ElseIf {\rm$\tail < \LM$} {

	\mbox{Push down $\xx_{\tail+1}$ to $\xx_{\LM}$ with }\raisebox{-.9em}{\includegraphics[width=.06\textwidth]{tst.26}}\;
	
	{\bf InvTFT}$\brac{\head,\, \tail,\, \LM,\, s+1}$\;
	
	\mbox{Push up $\xx_\head$ to $\xx_\LM$ with }\raisebox{-.9em}{\includegraphics[width=.06\textwidth]{tst.23}}\;

}

\end{algorithm}

%% file: TFT.bbl
\begin{thebibliography}{1}

\bibitem{AFCA}
K.~O. Geddes, S.~R. Czapor, and G.~Labahn.
\newblock {\em Algorithms for Computer Algebra}.
\newblock Kluwer Academic Publishers, 1992.

\bibitem{InPlaceTFT}
David Harvey and Daniel~S. Roche.
\newblock An in-place truncated fourier transform and applications to
  polynomial multiplication.
\newblock In {\em Proceedings of the 2010 International Symposium on Symbolic
  and Algebraic Computation}, ISSAC '10, pages 325--329, New York, NY, USA,
  2010. ACM.

\bibitem{SigProcessing}
H.V. Sorensen and C.S. Burrus.
\newblock Efficient computation of the dft with only a subset of input or
  output points.
\newblock {\em Signal Processing, IEEE Transactions on}, 41(3):1184 --1200, mar
  1993.

\bibitem{TFT2}
J.~van~der Hoeven.
\newblock Notes on the {Truncated} {Fourier} {Transform}.
\newblock Technical report, Universit\'e Paris-Sud, Orsay, France, 2008.

\bibitem{TFT1}
Joris van~der Hoeven.
\newblock The truncated fourier transform and applications.
\newblock In {\em ISSAC '04: Proceedings of the 2004 international symposium on
  Symbolic and algebraic computation}, pages 290--296, New York, NY, USA, 2004.
  ACM.

\end{thebibliography}
